\documentclass[11pt]{llncs}
\usepackage{makeidx}  
\usepackage{amsfonts}
\begin{document}
\mainmatter              
\title{Fast Approximate Matrix Multiplication by Solving Linear Systems}
%
%
\author{Shiva Manne\inst{1} \and Manjish Pal\inst{2}}
\institute{Birla Institute of Technology, Pilani, India \\
\email{manneshiva@gmail.com},
\and
National Institue of Technology Meghalaya, Shillong, India \\
\email{manjishster@gmail.com}}

\maketitle              

\begin{abstract}
In this paper, we present novel deterministic algorithms for multiplying two $n \times n$ matrices approximately. Given two matrices $A,B$ we return a matrix $C'$ which is an \emph{approximation} to $C = AB$. We consider the notion of approximate matrix multiplication in which the objective is to make the Frobenius norm of the error matrix $C-C'$ arbitrarily small. Our main contribution is to first reduce the matrix multiplication problem to solving a set of linear equations and then use standard techniques to find an approximate solution to that system in $\tilde{O}(n^2)$ time. To the best of our knowledge this the first examination into designing quadratic time deterministic algorithms for approximate matrix multiplication which guarantee arbitrarily low \emph{absolute error} w.r.t. Frobenius norm.
\end{abstract}
\section{Introduction}
The multiplication of two $n \times n$ matrices is one of the most basic problems in computer science. Significant effort has been devoted to finding efficient approaches to bounding the exponent ($\omega$) of matrix multiplication. The naive algorithm computes the product of two matrices in $O(n^3)$ operations. Strassen \cite{S1} in his seminal paper was the first to notice that the cubic algorithm was suboptimal.  He proved that multiplication could be achieved with complexity $O(n^{2.81})$. Following this, there was a sequence of works that have improved the exponent. Some of them were small improvements but they represent big conceptual advances. In a breakthrough work Coppersmith and Winograd \cite{CW} gave a very involved algorithm that proved that $\omega < 2.3754$. This bound has been further improved by Stothers to 2.373 \cite{DS}, by Williams to 2.3728642 \cite{W} and most recently by Francois Le Gall to 2.3728639 \cite{G}. These techniques in general can be looked as an attempt to bound the \emph{value} and \emph{border rank} of a trilinear form obtained by taking tensor powers of another trilinear form \cite{W}. In a different line of work Cohn, Umans \cite{CU} and later Cohn, Kleinberg, Szegedy and Umans \cite{CKZU} used a group theoretic construction to re-derive the bound of Coppersmith and Winograd. They also state certain conjectures whose truth implies that the exponent of matrix multiplication would be 2. Further work along this line has been done in \cite{ASU}. Recently, Iwen and Spencer \cite{IS} present a class of matrices whose product can be computed by a deterministic algorithm in time $O(n^{2+\epsilon})$. In addition to this there has been a plethora of work on sparse matrix multiplication \cite{AM,P}.

\subsection{What is Approximate Matrix Multiplication ?}
In this paper, we deal with the notion of approximate matrix multiplication i.e. given the input matrices $A,B$ we intend to return a matrix $C'$ which is an approximation to $C = AB$. The first algorithm for approximate matrix multiplication was proposed by Cohen and Lewis \cite{CL} that is based on sampling. For input matrices with nonnegative entries they show a concentration around the estimate for individual entries in the product matrix with high probability. Later Drineas et al \cite{DK,DKM} have introduced Monte-Carlo algorithms to bound the relative error w.r.t. Frobenius norm i.e. they bound the quantity, $\frac{\|C-C'\|_F}{\|A\|_F\|B\|_F}$ . Their algorithm simply samples $s$ columns from $A$ to form an $n \times s$ matrix $A'$ and $s$ rows from $B$ to form an $s \times n$ matrix and finally returns $A'B'$ as an approximate answer. If $s = O(1)$, the running time of the algorithm is $O(n^2)$. This algorithm is able to guarantee that the relative error can be bounded by $O(\frac{1}{\sqrt{s}})$ with high probability. In fact the algorithm of Cohen and Lewis can be looked as obtaining a bound of $O(\frac{1}{\sqrt{c}})$  where $n^2c$ samples are being used. In a different line of work Magen and Zousias \cite{MZ} have designed algorithms to approximate matrix multiplication guaranteeing arbitrarily low relative error w.r.t the spectral norm i.e. they form sketches of $A, B$ namely $\tilde{A}, \tilde{B}$ such that $\frac{\|C-\tilde{A}\tilde{B}\|_2}{\|A\|_2\|B\|_2}$ can be made arbitrarily small. Their algorithm is also randomized.  Recently, Pagh \cite{P} came up with a new randomized algorithm. Instead of sketching the input matrices and then multiplying the resulting smaller matrices,  the product is treated as a stream of outer products and each outer product is sketched. Using Fast Fourier Transform \cite{BC}, Pagh shows how to use the Count-Sketch algorithm to an outer product \cite{P}.
This algorithm has been derandomized by Kutzkov in \cite{K}.  In this work, we introduce the notion of designing deterministic quadratic time matrix multiplication algorithm that achieves arbitrarily low \emph{absolute error} w.r.t. Frobenius norm. Our technique is completely new as it uses tools from the theory of solving Linear Systems.

\section{Basic Idea}
Given two $n \times n$ matrices $A$ and $B$, consider the equation $AB = C$ where the objective is to find $C$. Multiplying both sides of the equation by a chosen vector $v = [v_1,v_2 \dots v_n]^{T}$, we get $ABv = Cv$, which can be rewritten as $Cv = u$ where computing $u$ can be done in $O(n^2)$ time. Thus we are reduced to solving the system $Vc = u$, where $c$ is the $n^2 \times 1$ variable vector and $V$ is the $n \times n^2$ matrix whose $i^{th}$ row is the $n^2$ dimensional vector whose first $n(i-1)$ entries are 0's, next $n$ entries are $v_1,v_2, \dots, v_n$ and the remaining entries are zeroes.  In what follows we present algorithms that use this property to find an approximate solution to the system of linear equations. Our main contribution is the following theorem.

\begin{theorem}
Given two $n \times n$ matrices $A,B$ such that the magnitude of each entry is atmost $M = O(1)$ and an $\delta > 0$, there exists an algorithm that runs in $O(n^2 \log \frac{1}{\delta})$ time and returns a matrix $C'$ such that the Frobenius norm of the matrix $C' - AB$ is atmost $\delta$. 
\end{theorem}

\section{Preliminaries}\label{Pre}
In this section we describe the basic linear algebraic preliminaries needed for our analysis.
\subsection{Vector Norms} 
\begin{definition}
A \emph{vector norm} is a mapping from $\|\cdot\| :\mathbb{C}^{n} \rightarrow \mathbb{C}$ that
satisfies the following
\begin{itemize}
\item $|\vec{x}| > 0$ for all $\vec{x} \neq 0. (|0| = 0)$   

\item $|\alpha \vec{x}| = |\alpha||\vec{x}|$. 

\item $| \vec{x} + \vec{y} | \leq |\vec{x}| + |\vec{y}|$. 
\end{itemize}
\end{definition}

Thus norm is a measure of of the length of a vector. Following are examples of some vector norms 
\begin{itemize}
\item[(a)] $l_p$ norm: $|\vec{x}|_p = (\sum_{i} |x_i|^p)^{1/p}$ for $1 \leq p \leq \infty$
\item[(b)] $l_\infty$ norm: $|\vec{x}|_\infty = \max_{i} |x_i|$.
\end{itemize}

\subsection{Matrix Norms}
\begin{definition}
A \emph{matrix norm} is a mapping from $\|\cdot\| :\mathbb{C}^{n \times n} \rightarrow \mathbb{C}$ that
satisfies the following 
\begin{itemize}
\item $\|A\| \geq 0$ for all $A \in \mathbb{C}^{n \times n}$ and $\|A\| = 0$ iff $A = 0$.
\item $\|\alpha A\| = |\alpha|\|A\|$ for all $A,B \in \mathbb{C}^{n \times n}$.
\item $\|A + B\| \leq \|A\| + \|B\|$ for all $A,B \in \mathbb{C}^{n \times n}$.
\item $\|AB\| \leq \|A\|\|B\|$ for all $A,B \in \mathbb{C}^{n \times n}$. 
\end{itemize}
\end{definition}

Notice that the first three properties states that $\|\cdot\|$ is a vector norm on $\mathbb{C}^{n \times n}$,
the last property only makes sense if $A$ and $B$ are matrices. Examples of matrix norms are 
\begin{itemize}
\item[(a)] Maximum norm: $\|A\|_{max} = \max |a_{ij}|$.
\item[(b)] Infinity norm: $\|A\|_{\infty} = \max_{1\leq i\leq n} \sum_{j=1}^n|a_{ij}|$.
\item[(b)] Frobenius norm: $\|A\|_{F} = \sqrt{\sum_{i,j} a_{ij}^2}$.  
\item[(c)] $p$-operator norm: $\|A\|_{P} = \max_{|\vec{x}|_p = 1} |A\vec{x}|_p$, where $|\cdot|_p : \mathbb{C}^{n} \rightarrow n$ is the vector $p$-norm and $1 \leq p \leq \infty$.  
\end{itemize}

\subsubsection{Induced Norms}

\begin{definition}
Given a vector norm $|\cdot|_v$ on $\mathbb{C}^n$ we can define \emph{the induced norm} $\|\cdot\|$ on $\mathbb{C}^{n \times n}$ as 
\begin{eqnarray*}
\|A\| = \max_{\vec{x} \neq 0}\frac{|A\vec{x}|}{|\vec{x}|}
\end{eqnarray*}
for all $A \in \mathbb{C}^{n \times n}$
\end{definition}
It can be verified that an induced norm is indeed a matrix norm.  For example, $p$-operator norm is an induced norm for all
$1 \leq p \leq \infty$. 
Another useful notion in the context of matrix norms is that of \emph{spectral radius}, 
\begin{definition}
The \emph{spectral norm} of a matrix $A$ denoted by $\rho(A)$ is defined as
\begin{eqnarray*}
\rho(A) = \max_{1 \leq i \leq n} |\lambda_i|
\end{eqnarray*}
where $\lambda_1,\lambda_2 \dots \lambda_n$ are the eigenvalues of $A$.  
\end{definition}
The following is an important fact connecting spectral norm and induced norms which is going to be of use for us later
\begin{theorem}
For any induced norm $\|\cdot\|$ defined on $\mathbb{C}^{n \times n}$, the following holds for every matrix $A$,
\begin{eqnarray*}
\rho(A) \leq \|A\|.
\end{eqnarray*}
\end{theorem}
Proof of this can be found in any standard text book on matrix analysis \cite{HJ}.

\section{Using Positive Definite System}
In this section we describe a deterministic algorithm to find an approximate solution to the system mentioned in Section. The idea is to pre-multiply both sides of equation $Vc = u$ by $V^{T}$ to get $V^{T}V c = V^{T}u$, where $V^{T}V$ is a positive semi-definite matrix. Thus we have converted the system into form $Ax = y$ where $A$ is a symmetric positive semi-definite matrix. Unfortunately, there are no methods known to find an approximate solution to a system in this form but if $A$ is a symmetric positive definite then there are iterative methods to find an approximate solution to the system. In the remaining of this section we first describe an iterative method to solve positive definite system and discuss how $A$ can be converted to a positive definite matrix and then apply iterative methods to solve the system.

\subsection{Steepest Descent Method for solving Positive Definite Systems}
Solving positive definite systems is a topic of deep interest. A good reference for iterative methods for solving positive definite systems is \cite{GL}. Steepest Descent is a general iterative method for finding local minima of a function $f$. In the context of Linear systems it is useful because the solution to $Ax = b$ is the $x$ that minimizes the quadratic form $f(x) = \frac{1}{2}x^{T}Ax - bx + c$. Given a current estimate $x_i$, the gradient $\nabla f(x_i)$ or more precisely, its negative gives the direction in which $f$ is decreasing most rapidly. Hence, one would expect that taking a step in this direction should bring us closer to the minimum we seek. We will see that in order to guarantee convergence of this method $A$ has to be positive definite. Let $x$ denote the actual minimizer, $x_i$ denote our $i^{th}$ estimate, and
\begin{eqnarray*}
e_i &=& x - x_i \\
r_i &=& b - Ax_i = Ae_i
\end{eqnarray*}
The update rule is 
\begin{eqnarray*}
x_{i+1} &=& x_i + \alpha_ir_i
\end{eqnarray*}
where $\alpha_i$ is chosen such that $r_{i+1}r_i = 0$. Using simple algebraic manipulations the value of 
$\alpha_i$ turns out to be $\frac{r_i^{T}r_i}{r_i^{T}Ar_i}$. 

\subsubsection{Convergence}
It can be shows that the steepest descent method converges, i.e. the 2-norm of the error vector is atmost $\rho$ in $O(\kappa \log \frac{1}{\rho})$ steps, where $\kappa = \frac{\lambda_{max}(A)}{\lambda_{min}(A)}$.
We can write $e_{i+1}$ (the error vector at the ${i+1}^{th}$ iteration) in terms of $e_i$ as follows:
\begin{eqnarray*}
e_{i+1} &=& x_1 + \alpha_i(b - Ax_i) - x* \\
&=& (I - \alpha_i A) e_i
\end{eqnarray*}
It can be proven that if we analyse this iteration by taking a fixed
value of $\alpha_i = \frac{2}{\lambda_{max} + \lambda_{min}}$, which is worse than
the value of $\alpha_i$ we actually choose, the 2-norm of the error vector $|e_i| \leq \rho$
in $O(\kappa \log \frac{1}{\rho})$ iterations \cite{l1}.

\subsection{Perturbation of $A$}
Recall that $A = V^{T}V$ which has zero eigenvalues and hence is not positive definite. We perform a perturbation of $A$ to $\hat{A} = A + \epsilon I$ where $\epsilon > 0$ will be fixed later. Notice that $A$ is a block diagonal matrix with $A'$ as the $n \times n$ matrix that appears in the main diagonal $n$ times where $A'$ is the outerproduct $vv^{T}$. The following holds for $\hat{A}$.

\begin{lemma}
$\hat{A}$ is positive definite.
\end{lemma}
\begin{proof}
We need to show that for any non-zero vector $x$, $x^{T}Ax > 0$ which is equivalent to proving $x^{T}Ax + \epsilon x^{T}Ax > 0$. Since $A$ is positive semi-definite and $\epsilon > 0$, it is true. 
\end{proof}

\begin{lemma}
If $\lambda$ is the maximum eigenvalue of $A$, then $\lambda = \sum_{i=1}^n v_i^2$.
\end{lemma}
\begin{proof}
Since $A$ has just one non zero eigenvalue which is also the only non-zero eigenvalue of $A'$.
To find the eigenvalue of $A'$ we consider the characteristic polynomial of $A'$ which is 
\begin{eqnarray*}
\det(\lambda I - A') &=& \det(\lambda I - vv^{T}) \\
&=& (1 - \frac{1}{\lambda}v^{T}v) |\lambda| 
\end{eqnarray*}
where the last line follows from a result in \cite{ZD}. Thus to find the to find the eigen-value of $vv^{T}$ we have to solve for $(1 - \frac{1}{\lambda}v^{T}v) |\lambda| = 0$ that gives $\lambda = v^{T}v =  \sum_{i=1}^n v_i^2 $
\end{proof}

\begin{lemma}
Let $\lambda_{max}(\hat{A})(\lambda_{min}\hat{A})$ be the maximum(minimum) eigenvalue of $\hat{A}$, then $\frac{\lambda_{max}(\hat{A})}{\lambda_{min}(\hat{A})} \leq 1 + \frac{\lambda}{\epsilon}$ where $\lambda = \sum_{i=1}^n v_i^2$.
\end{lemma}
\begin{proof}
Since $\hat{A}$ is block diagonal, $\hat{A}^{-1}$ is also block diagonal with $(A' + \epsilon I)^{-1}$ repeated in the main diagonal $n$ times. By Sherman-Morrison formula,
\begin{eqnarray*}
(vv^{T} + \epsilon I)^{-1} &=& (\epsilon I)^{-1} - \frac{(\epsilon I)^{-1}vv^{T}(\epsilon I)^{-1}}{1 + v^{T}(\epsilon I)v}
\\
&=& \frac{1}{\epsilon}\cdot I - \frac{1}{\epsilon^{2}} \cdot \frac{A'}{1 + \epsilon \sum_{i=1}^{n} v_i^2}
\end{eqnarray*} 
By lemma, if $\lambda_i$ is an eigenvalue of $A$ then $\frac{1}{\lambda}$ is an eigenvalue of $A^{-1}$. Also from Section \ref{Pre} for any induced norm $\|\cdot\|$,
\begin{eqnarray*}
\lambda_{max}(\hat{A}) \leq ||\hat{A}||_{2} \\
\lambda_{max}({\hat{A}^{-1}}) \leq ||\hat{A}^{-1}||_{2} \\
\frac{1}{\lambda_{min}(\hat{A})} \leq ||\hat{A}^{-1}||_{2}
\end{eqnarray*}
Since $\hat{A}$ is symmetric $||\hat{A}||_{2} = \sqrt{\lambda_{max}(\hat{A}^{T}\hat{A})} = \lambda_{max}(\hat{A})$ and since $\hat{A}^{-1}$ is also symmetric, $||\hat{A}^{-1}||_{2} =  \lambda_{max}(\hat{A}^{-1}) = \frac{1}{\lambda_{min}(\hat{A})}$. By the definition of $\hat{A}$, $\lambda_{max}(\hat{A}) = \lambda + \epsilon$ where $\lambda = \sum_{i=1}^n v_i^2$ and $\lambda_{min}(\hat{A}) = \epsilon$. Thus, $\frac{\lambda_{max}(\hat{A})}{\lambda_{min}(\hat{A})} \leq 1 + \frac{\lambda}{\epsilon}$.
\end{proof}

\subsection{Error Analysis}
Let $x'$ be the output of the algorithm, $x''$ be the a solution to the equation $(\hat{A}+ \epsilon I) x = y$ and $x'''$ be the solution to the original equation $Ax = y$.  Because of the guarantees of the algorithm we can assume that $|x' - x''| \leq \rho$. Our aim is to bound the 2-norm of the error vector $x' - x'''$. First we derive upper and lower bounds on $|x''|$.
\begin{eqnarray*}
|y| \leq ||(\hat{A}+ \epsilon I)||_2 |x''| \\
|x''| \geq \frac{||(\hat{A}+ \epsilon I)||_2}{|y|} 
\end{eqnarray*}
and
\begin{eqnarray*}
|x''| \leq ||(\hat{A}+ \epsilon I)^{-1}||_2 |y|
\end{eqnarray*}
We now prove a result that shows that the norm of the error vector $|x'-x'''|$ can be made arbitrarily small.

\begin{lemma}
$|x' - x'''| \leq \delta$ where $\delta > 0$ can be made arbitrarily small.
\end{lemma} 
\begin{proof}
Using the fact that $(\hat{A} + \epsilon I)x'' = \hat{A} x'''$, we have $|\hat{A} x'''| = |(\hat{A} + \epsilon I)x'|$. Now
\begin{eqnarray*}
|\hat{A}x'''| &=& \sqrt{\sum_{j=1}^{n}(nv^{2} \sum_{i=1}^n c_{ji})^2}  = \sqrt{n}v \sqrt{{\sum_{j=1}^{n}(\sum_{i=1}^n c_{ji})^2}} \\
&=& \sqrt{n}v \sqrt{{\sum_{i,j=1}^{n} c_{ji}^2} + 2 \sum_{j=1}^n\sum_{k_1<k_2}^{n} c_{jk_1} c_{jk_2}} \\
&\geq& \sqrt{n}v\sqrt{{\sum_{i,j=1}^{n} c_{ji}^2}} = \sqrt{n}v|x'''|
\end{eqnarray*} 
Thus from above
\begin{eqnarray*}
\sqrt{n}v|x'''| &\leq& |(\hat{A} + \epsilon I)x''|  \leq \|(\hat{A} + \epsilon I)\|_2 |x''| \\
&\leq&   \|(\hat{A} + \epsilon I)\|_2 \|(\hat{A} + \epsilon I)^{-1}\|_2 |y| \\
& = & \frac{\lambda + \epsilon}{\epsilon} \sqrt{\sum_{i,j}^n v_i^2\left(\sum_{k=1}^n c_{ik}v_j\right)^2} \\
|x'''| &\leq& \frac{1}{\sqrt{n}v} \cdot \frac{\lambda + \epsilon}{\epsilon} \sqrt{\sum_{i,j}^n v_i^2\left(\sum_{k=1}^n c_{ik}v_j\right)^2}  \\
&\leq& \frac{1}{n^2\sqrt{n}} \cdot\frac{\lambda + \epsilon}{\epsilon} \cdot M'.  
\end{eqnarray*} 
where $M' = O(n)$ is the maximum row sum of the resultant matrix $C$ and  $\epsilon = v_i = v = \frac{1}{n^3}$ for all $i = 1, 2 \dots, n$. From the above analysis we have,
\begin{eqnarray*}
|x' - x'''| &\leq& |x' - x''| + |x'' - x'''| \leq \rho + |x''| + |x'''| \\ 
&\leq& \rho + \|(\hat{A} + \epsilon I)^{-1}\|_2 |y| +  \frac{1}{\sqrt{n}} \cdot\frac{\lambda + \epsilon}{\epsilon} \cdot M' \\
&=& \rho +  \frac{M}{n^5\epsilon} + \frac{1}{n^2\sqrt{n}} \cdot\frac{\lambda + \epsilon}{\epsilon} \cdot M' \leq \delta. 
\end{eqnarray*}
where $\delta$ is a constant greater than but arbitrarily close to $\rho$ (say 1.001 $\rho$). 
\end{proof}

\subsection{Running Time Analysis}
The running time analysis is simple, the number of iterations of the steepest descent method is $O(\frac{\lambda_{max}(\hat{A})}{\lambda_{min}(\hat{A})} \log \frac{1}{\rho})$ which according to our choice of $\epsilon$ and $v_i'$s is
$O(\log \frac{1}{\rho})$. Every iteration involves multiplying $\hat{A}$ with a vector $\vec{x}$ which despite the fact
that $\hat{A}$ is an $n^2 \times n^2$ matrix can be done in $O(n^2)$ time. The reason being $\hat{A} = A +  \epsilon I$ and both $A\vec{x}$ and $I\vec{x}$ can be computed in $O(n^2)$ time. Note that the matrix $\hat{A}$ is never stored as it is.  

\section{Conclusion}
In this paper we have introduced a new technique of multiplying two matrices approximately by solving a set of linear equations. By using standard methods for solving some specific linear systems namely positive definite systems we have been able to design deterministic algorithms to ensure arbitrarily small absolute error between our answer and the actual product. Such a result to the best of our knowledge is the first of its kind in the context of approximate matrix multiplication. We suspect that this technique will find further applications in the problem of exact matrix multiplication. 


%


\begin{thebibliography}{99}
%

\bibitem[1]{AM}
D. Achiloptas and F. McSherry, Fast Computation of Low Rank Approximations, Proceedings of the 33rd Annual Symposium on Theory of Computing(2001).

\bibitem[2]{ASU}
	N. Alon, A. Shpilka  and C. Umans, On sunflowers and matrix multiplication, ECCC TR11-067, 18 (2011).

\bibitem[3]{BC}	U. Baum and Mc. Clausen, Fast Fourier Transforms, SpektrumAkademischerVerlag (1993).

\bibitem[4]{CL}
	E. Cohen and D. D. Lewis, Approximating matrix Multiplication for Pattern recognition tasks, Journal of Algorithms, 30(2): 211-252 (1999).

\bibitem[5]{CKZU}
	H. Cohn, R. D. Kleinberg, B. Szegedy, and C. Umans, Group-theoretic algorithms for matrix multiplication, in Proceedings of the 46th Annual FOCS, pp. 379–388 (2005).

\bibitem[6]{CU}
	H. Cohn and C. Umans, A group-theoretic approach to fast matrix multiplication, in Proceedings of the 44th Annual FOCS, pp. 438–449 (2003).

\bibitem[7]{CW}	D. Coppersmith and S. Winograd, Matrix multiplication via arithmetic progressions, J. Symbolic Comput., 9, pp. 251–280 (1990).

\bibitem[8]{DS}	A. Davie and A. J. Stothers, Improved bound for complexity of matrix multiplication, Proceedings of the Royal Society of Edinburgh, Section: A Mathematics, 143:pp. 351–369 (2013).

\bibitem[9]{DK}
P. Drineas and R. Kannan, Fast Monte-Carlo algorithms for approximate matrix multiplication, in Proceedings of the 42nd Annual IEEE Symposium on Foundations of Computer Science, pp. 452–459 (2001).

\bibitem[10]{DK1}	P. Drineas and R. Kannan, Pass efficient algorithms for approximating large matrices, in Proceedings of the 14th Annual ACM-SIAM Symposium on Discrete Algorithms, pp. 223–232 (2003).

\bibitem[11]{DKM}	P. Drineas, R. Kannan, and M. W. Mahoney, Fast Monte Carlo Algorithms for Matrices I: Approximating Matrix Multiplication, SIAM J. Comput, 36, pp. 132-157 (2006).

\bibitem[12]{DMM}	P. Drineas, M. W. Mahoney and S. Muthukrishnan, Sub-space sampling and relative-error matrix approximation: Column-based methods, in proc. Of the 10th RANDOM (2006).

\bibitem[13]{FKV}	A. Frieze, R. Kannan, and S. Vempala, Fast Monte-Carlo algorithms for finding low-rank approximations, in Proceedings of the 39th Annual IEEE Symposium on Foundations of Computer Science, pp. 370–378 (1998).

\bibitem[14]{G}	Francois Le Gall, Powers of tensors and fast matrix multiplication, Proceedings of the 39th International Symposium on Symbolic and Algebraic Computation (2014).

\bibitem[15]{GL}	G. H. Golub and C. F. Van Loan, Matrix Computations, Johns Hopkins University Press, Baltimore, MD (1989).

\bibitem[16]{HMT}	N. Halko, P.G. Martinsson and J.A. Tropp, Finding Structure with Randomness: Probabilistic Algorithms for Constructing Approximate Matrix Decompositions, SIAM Review 53, pp. 217-288 (2011).

\bibitem[17]{IS} M. A. Iwen and C. V. Spencer. A note on compressed sensing and the complexity of
matrix multiplication. Inf. Process. Lett, 109(10):468–471, 2009

\bibitem[18]{K}		K. Kutzkov, Deterministic algorithms for skewed matrix products, CoRR abs/1209.4508 (2012).

\bibitem[19]{MZ}	A. Magen and A. Zouzias, Low Rank Matrix-valued Chernoff Bounds and Approximate Matrix Multiplication, SODA: 1422-1436 (2011).

\bibitem[20]{P}	R. Pagh, Compressed Matrix Multiplication, TOCT 5(3): 9 (2013).

\bibitem[21]{PN} 	V. Pan, How Can We Speed Up Matrix Multiplication?,  SIAM Review Volume 26, No. 3, pp. 393-415 (1984).

\bibitem[22]{SS}	G. W. Stewart and J. G. Sun, Matrix Perturbation Theory, Academic Press, New York (1990).

\bibitem[23]{S1} V. Strassen, Gaussian elimination is not optimal, Numer. Math., 14, pp. 354–356 (1969).

\bibitem[24]{S2} V. Strassen, Relative bilinear complexity and matrix multiplication, J. ReineAngew. Math., 375/376:406–443 (1987).

\bibitem[25]{T} J. Takche, Complexities of Special Matrix Multiplication Problems, Comput. Math. Applic. Vol. 15, No. 12, pp. 977-989 (1988).

\bibitem[26]{W} V.V. Williams,  Multiplying matrices faster than Coppersmith-Winograd, In Proceedings of the 44th Symposium on Theory of Computing, STOC '12, ACM, 887–898 (2012).

\bibitem[27]{HJ} R. A. Horn and C. R. Johnson, Matrix Analysis, Cambridge University Press (2013).

\bibitem[28]{ZD}	A. Zhou and J. Ding, Eigenvalues of rank-one updated matrices with some applications, Applied Mathematics Letters vol. 2 : issue 12, pp. 1223-1226 (2007).

\bibitem[29]{l1} http://www.cs.berkeley.edu/satishr/cs270/sp11/rough-notes/Linear-Equations.pdf


\end{thebibliography}
\end{document}